\documentclass[11pt]{article}

\usepackage{amsmath, amsthm, amssymb, graphicx, enumerate, fullpage, color}
\usepackage{pseudocode, framed}
\usepackage{algorithmicx, algpseudocode}

\theoremstyle{plain}
\newtheorem{theorem}{Theorem}[section]
\newtheorem{prop}[theorem]{Proposition}
\newtheorem{lemma}[theorem]{Lemma}
\newtheorem{cor}[theorem]{Corollary}

\theoremstyle{definition}
\newtheorem{defn}[theorem]{Definition}

\newcommand{\cP}{\mathcal{P}}
\newcommand{\cR}{\mathcal{R}}
\newcommand{\cT}{\mathcal{T}}
\newcommand{\minus}{\setminus}

\DeclareMathOperator{\diam}{diam}
\DeclareMathOperator{\dist}{dist}

\title{Local reconstructors and tolerant testers for connectivity and diameter}
\author{%
		Andrea Campagna%
	\thanks{This research was done while at IT University of Copenhagen and
	while visiting the Blavatnik School of Computer Science of Tel Aviv University.
	{\tt acam@itu.dk}.}
	\and
	Alan Guo%
	\thanks{CSAIL, Massachusetts Institute of Technology, Cambridge MA 02139.
	{\tt aguo@mit.edu}.
	Research supported in part by NSF grants CCF-0829672, CCF-1065125,
	CCF-6922462, and an NSF Graduate Research Fellowship.
	}
	\and
	Ronitt Rubinfeld%
	\thanks{CSAIL, Massachusetts Institute of Technology, Cambridge MA 02139
	and the Blavatnik School of Computer Science, Tel Aviv University.
	{\tt ronitt@csail.mit.edu}.
	Research supported by NSF grant 1065125 
and the Israel Science Foundation grant no. 1147/09.}
}
\date{August 12, 2012}

\newif\ifabstract
\abstracttrue
\abstractfalse
\newif\iffull
\ifabstract \fullfalse \else \fulltrue \fi

\newcounter{section-preserve}
\newcounter{theorem-preserve}
\newcommand{\blank}[1]{}
\newtoks\magicAppendix
\magicAppendix={}
\newtoks\magictoks
\newif\iflater
\laterfalse
\long\def\later#1{\magictoks={#1}%
  \edef\magictodo{\noexpand\magicAppendix={\the\magicAppendix \par
    \the\magictoks%
  }}
  \magictodo}
\long\def\both#1{\magictoks={#1}%
  \edef\magictodo{\noexpand\magicAppendix={\the\magicAppendix \par
    \noexpand\setcounter{theorem-preserve}{\noexpand\arabic{theorem}}%
    \noexpand\setcounter{theorem}{\arabic{theorem}}%
    \noexpand\setcounter{section-preserve}{\noexpand\arabic{section}}%
    \noexpand\setcounter{section}{\arabic{section}}%
        \noexpand\let\noexpand\oldsection=\noexpand\thesection
        \noexpand\def\noexpand\thesection{\thesection}
        \noexpand\let\noexpand\oldlabel=\noexpand\label
        \noexpand\let\noexpand\label=\noexpand\blank
    \the\magictoks%
    \noexpand\setcounter{theorem}{\noexpand\arabic{theorem-preserve}}%
    \noexpand\setcounter{section}{\noexpand\arabic{section-preserve}}%
        \noexpand\let\noexpand\thesection=\noexpand\oldsection
        \noexpand\let\noexpand\label=\noexpand\oldlabel
  }}
  \magictodo
  \the\magictoks}
\def\magicappendix{\latertrue \the\magicAppendix}

\iffull
  \long\def\both#1{#1}
  \let\later=\both
  \def\magicappendix{}
\fi

\begin{document}

\maketitle
\thispagestyle{empty}

\begin{abstract}
A local property reconstructor for a graph property is an algorithm
which, given oracle access to the adjacency list of a graph that is ``close''
to having the property, provides oracle access to the adjacency matrix
of a ``correction'' of the graph, i.e. a graph which has the property
and is close to the given graph. For this model, we achieve local property
reconstructors for the properties of connectivity and $k$-connectivity in undirected graphs,
and the property of strong connectivity in directed graphs. Along the way, we
present a method of transforming a local reconstructor
(which acts as a ``adjacency matrix oracle'' for the corrected graph)
into an ``adjacency list oracle''. This allows us to recursively use our local
reconstructor for $(k-1)$-connectivity to obtain a local reconstructor
for $k$-connectivity.

We also extend this notion of local property
reconstruction to parametrized graph properties (for instance, having
diameter at most $D$ for some parameter $D$) and require that
the corrected graph has the property with parameter close to the original.
We obtain a local reconstructor for the low diameter property,
where if the original graph is close to having diameter $D$,
then the corrected graph has diameter roughly $2D$.

We also exploit a connection between local property reconstruction
and property testing, observed by Brakerski,
to obtain new tolerant property testers for all of the aforementioned properties.
Except for the one for connectivity, these are the first tolerant property testers for
these properties.

\end{abstract}

\clearpage
\pagenumbering{arabic}

\section{Introduction}

Suppose we are given a very large graph $G$ that is promised to be close to having a
property $\cP_s$. For example, $\cP_D$ might denote the property of having diameter at most
$D$.
Local reconstruction algorithms provide very fast query access to a
``corrected'' version of $G$.
That is, the local reconstruction algorithm should have in mind some
$\widetilde{G}$ which has the property $\cP_s$
and is also close to the original graph $G$.
The goal of the local reconstruction algorithm is to provide very fast query
access to the edges of $\widetilde{G}$ ---
that is, given a pair of vertices $u, v$ in $G$, the algorithm should
in sublinear time determine whether the edge $(u,v)$ is in $\widetilde{G}$.
We call such an algorithm a \emph{local reconstructor} for $\cP_s$.
It can be useful to relax the condition that $\widetilde{G}$ has property $\cP_s$
and only require that $\widetilde{G}$ has property $\cP_{\phi(s)}$
which contains $\cP_s$ but is possibly larger. For instance, if $\cP_D$ is the property of
having diameter at most $D$, we might only require $\widetilde{G}$ to have
property $\cP_{4D}$, i.e.\ having diameter at most $4D$.

In this paper we study local reconstruction algorithms for some of the most
basic problems in graph theory, namely connectivity in undirected graphs,
strong connectivity in directed graphs,
$k$-connectivity in undirected graphs,
and small diameter in undirected graphs.
Such algorithms might be used to efficiently repair 
connectivity or small diameter in graphs. These are common issues that have been
considered in various models  in wireless networks
and robotics (see for example \cite{DCR11, SJK08}).

Techniques for designing local reconstructors are often borrowed from property testing,
and as noted by Brakerski~\cite{B08} (as well as in this paper), reconstructors can be used
to design property testers.
Property testers have been studied extensively in the literature (see, for instance,
~\cite{RS96, GGR98, GR02} for some early works on the subject).
A property tester for a property $\cP$ takes as input
a graph $G$ and parameter $\epsilon$, and accepts with high probability if $G$
has $\cP$
and rejects with high probability if $G$ is $\epsilon$-far from having $\cP$. Similarly, a
tolerant tester
gets a graph $G$ and parameters $\epsilon_1 \le \epsilon_2$, and accepts
with high probability
if $G$ is $\epsilon_1$-close to having $\cP$ and rejects with high probability if $G$ is
$\epsilon_2$-far from having $\cP$.

\paragraph{Our results.}

Our specific results are the following.
\begin{itemize}

\item%
We give local reconstructors for the following properties:
connectivity in undirected graphs;
strong connectivity in directed graphs;
$k$-connectivity in undirected graphs;
having diameter at most $D$, for some diameter parameter $D$, in
undirected graphs.

\item%
We present a method of transforming our local reconstructor
for $k$-connectivity (which provides query access to the adjacency
\emph{matrix} of the corrected graph) into an algorithm which
provides query access to the adjacency \emph{list} of the corrected
graph (see Section~\ref{s:neighbor} for the case $k=1$). We note that this is
\emph{not} a black-box transformation, and is specific to our
the local reconstructors that we design.

\item%
We exploit a connection between local property reconstruction and
property testing, observed by Brakerski~\cite{B08}, which we generalize
to the setting of parametrized graph properties, in order to
obtain tolerant property testers for all of the above
graph properties (property testing notions will be defined shortly).

\end{itemize}

\paragraph{Our approach.}
Our techniques are simple, yet seem to be quite powerful given their simplicity.
For each of the above properties, the strategy for constructing a local
reconstructor is the same.
First, we designate a ``super-node''; then elect ``leader nodes''
from which we add edges to the super-node.
The main technique we use to elect leaders is to initially independently assign
a random rank to every node, and to declare a node a leader if it has the lowest rank
among all nodes within a small neighborhood. This is more useful than simply choosing
leaders at random, since sometimes we would like to guarantee we have a leader; for instance,
for connectivity, we need at least one leader in each connected component if we want
to \emph{guarantee} that the corrected graph is connected.

Brakerski~\cite{B08} gave a way to construct a tolerant tester from
a local reconstructor and property tester.
\iffull{
Brakerski's idea (see~\cite{B08}) behind the construction of the tolerant tester
is to attempt to use the local reconstructor on $G$ to get $\widetilde{G}$
which has the property and is close to $G$. If $G$ is indeed close to having
the property, then the attempt will be successful. If $G$ is far from having the property, then
either $\widetilde{G}$ will be far from having the property (which can be detected
by running the property tester on $\widetilde{G}$) or $G$ will be far from $\widetilde{G}$
(which can be detected via sampling).
}\fi
Since previous works (\cite{GR02}, \cite{BR02}, \cite{MR09},
 \cite{PR02})
give property testers for these properties we study, we obtain, as corollaries, that
these properties have tolerant testers.

\paragraph{Related work.}
\iffull{
The notion of locally reconstructing a data set was introduced under
the name \emph{local filter} in~\cite{ACCL04}, where the property
considered was monotonicity of sequences.
A model of local filters in which the requirements are strengthened is presented
in~\cite{SS08}, where the property of monotonicity is again considered.
}\fi %
%
A closely related work is~\cite{B08}, which introduces the concept
of local reconstruction under the name \emph{local restoring} and
also shows a special case of our relationship between local reconstructors,
property testers, and tolerant testers. 
\iffull{
Several other properties of graphs, functions, and geometric point sets
have been studied in the context of local reconstruction.
A local reconstruction algorithm for expander graphs is
given in~\cite{KPS08}. \cite{B08} gives local reconstructors
for the properties of bipartiteness and $\rho$-clique in the dense graph model,
as well as monotonicity. A more recent work (\cite{JR11}) concerns
locally reconstructing Lipschitz functions.
On the geometric side, \cite{CS06} studies local reconstruction
of convexity in two and three dimensions.
The problem of testing the properties for which we give local reconstructors
has been studied in multiple works---\cite{GR02} for connectivity,
\cite{BR02} for strong connectivity, \cite{GR02} and \cite{MR09} for $k$-connectivity,
and \cite{PR02} for diameter.
}\fi%

\paragraph{Organization.}
In Section~\ref{s:prelim}, we formally define our model and the notions
of local reconstructors, property testers, and tolerant testers, and
in Section~\ref{s:results} we
formally state our main results. In Section~\ref{s:conn} and~\ref{s:sconn},
we present our results for connectivity and strong connectivity.
Section~\ref{s:neighbor} serves as a brief interlude where we show
how our local reconstructor for connectivity can be modified to give
a neighbor oracle for the corrected graph $G'$. This procedure
will then be used in Sections~\ref{s:kconn} and~\ref{s:diam}, in which
we present our results for $k$-connectivity and small diameter respectively. 

\section{Preliminaries}\label{s:prelim}

We adopt the general sparse model of graphs as presented in~\cite{PR02},
i.e.\ the graph is given as an adjacency list, and
a query for a vertex $v$ is either its degree, or an index $i$ on which
the $i$-th neighbor of $v$ is returned (with respect to the representation of the
neighbors as an ordered list).
We assume there is some upper bound $m$ on the number
of edges of the graphs we work with, and distances are measured
according to this, i.e.\ if $k$ is the minimum number of edge deletions and
insertions necessary to change one graph to the other, then their distance
is $k/m$. We assume $m = \Omega(n)$ where $n$ is the number of vertices
in the graph.

\later{

\begin{defn}[\cite{PR02}]
The \emph{distance between two graphs $G_1,G_2$}, denoted
$\dist(G_1,G_2)$, is equal to the number of unordered pairs
$(u,v)$ such that $(u,v)$ is an edge in one graph but not in the other,
divided by $m$.
A \emph{property} is a subset of graphs.
Throughout this paper we say that a graph has property $\cP$ if it
is contained in the subset $\cP$.
The \emph{distance between a graph $G$ and a property $\cP$},
denoted $\dist(G,\cP)$, is equal to $\dist(G,\cP) = \min_{G' \in \cP} \dist(G,G')$.
If $\dist(G,\cP) \le \epsilon$, then $G$ is \emph{$\epsilon$-close to $\cP$},
otherwise $G$ is \emph{$\epsilon$-far from $\cP$}.
\end{defn}
}

\paragraph{Parametrized properties.}
Our result relating local reconstructors to tolerant testers (Theorem~\ref{t:tolerant_param})
generalizes a result of Brakerski (\cite{B08}) to parametrized properties.
A \emph{parametrized property} $\cP_s$ is a property belonging to a family
$\{\cP_s\}_s$ of properties parametrized by some parameter $s$. For
example, the property $\cP_D$ of having diameter at most $D$ is a parametrized
property, with the diameter $D$ as the parameter.

\subsection{Local reconstructors}

\begin{defn}
For an undirected graph $G = (V,E)$, the \emph{neighbor set of $v \in V$}
is the set
$
N_G(v) = \{u \in V \mid (u,v) \in E\}.
$
For a directed graph $G = (V,E)$, the \emph{in-neighbor set and
out-neighbor set of $v \in V$} are respectively
$
N^{\rm in}_G(v) = \{u \in V \mid (u,v) \in E\}
$
and
$
N^{\rm out}_G(v) = \{u \in V \mid (v,u) \in E\}
$
\end{defn}

\begin{defn}[Neighbor and edge oracles]
A \emph{neighbor oracle} for a graph $G$ is an algorithm which, given query
$v \in V$ and either query $\deg$ or $i$, returns $\deg(v)$ or
the $i$-th neighbor of $v$ (with respect to some fixed ordering of the neighbor set)
in $O(1)$ time.
An \emph{edge oracle for $G$} is an algorithm $E_G$ which
returns in $O(1)$ time $E_G(u,v) = 1$ if $(u,v) \in E$
and $E_G(u,v) = 0$ otherwise.
\end{defn}

One can use a neighbor oracle to implement an edge oracle
with query complexity $\deg(v)$, since given
a query pair $(u,v)$, one can check if $u \in N_G(v)$ with $\deg(v)$ queries
to the neighbor oracle. We now formally define
local reconstructors. Roughly speaking, a local reconstructor uses a
neighbor oracle for $G$, which is close to $\cP$, to implement an edge oracle for
$\widetilde{G} \in \cP$ which is close to $G$.

\begin{defn}
Let $\epsilon_1,\epsilon_2, \delta > 0$ and let $\phi : \mathbb{N} \to \mathbb{N}$.
An \emph{$(\epsilon_1,\epsilon_2,\delta,\phi(\cdot))$-local reconstructor (LR)}
$\cR$ for a parametrized graph property $\cP_s$ is a randomized algorithm
with access to a neighbor oracle of a graph $G$ that is $\epsilon_1$-close to
$\cP_s$, which satisfies the following:
\begin{itemize}
\item%
$\cR$ makes $o(m)$ queries to the neighbor oracle for $G$
per query to $\cR$

\item%
There exists $\widetilde{G} \in \cP_{\phi(s)}$ with $\dist(G,\widetilde{G}) \le \epsilon_2$
such that $\cR$ is an edge oracle for $\widetilde{G}$,
with probability at least $1-\delta$
(over the coin tosses of $\cR$)

\end{itemize}
An \emph{$(\epsilon_1,\epsilon_2,\delta)$-local reconstructor} for a non-parametrized
graph property $\cP$ is simply a $(\epsilon_1,\epsilon_2,\delta,\phi(\cdot))$-local
reconstructor where $\cP$ is viewed as the only property in its parametrized family
and $\phi$ is the identity function.
The \emph{query complexity} of the local reconstructor is the number of queries
$\cR$ makes to the neighbor oracle for $G$ on any query $(u,v)$.
We note that this definition differs from that of~\cite{JR11} because even if
$G \in \cP$, the reconstructed graph $\widetilde{G}$ may not equal $G$ in general.
\end{defn}

\subsection{Tolerant testers}

Tolerant testers (see~\cite{PRR04}) are a generalization of property testers where the tester
may accept if the input is close enough to having the property, where
for property testers ``close enough'' means ``distance zero''.

\begin{defn}
Let $\epsilon_1,\epsilon_2 > 0$ and let $\phi : \mathbb{N} \to \mathbb{N}$.
An \emph{$(\epsilon_1,\epsilon_2,\phi(\cdot))$-tolerant tester}
$\cT$ for a parametrized graph property $\cP_s$ is a randomized algorithm
with query access to a neighbor oracle of an input graph $G$ that satisfies the following:
\begin{itemize}
\item%
$\cT$ makes $o(m)$ queries to the neighbor oracle for $G$

\item%
If $G$ is $\epsilon_1$-close to $\cP_s$, then
$\Pr[\cT {\rm ~accepts~}] \ge \frac23$

\item%
If $G$ is $\epsilon_2$-far from $\cP_{\phi(s)}$, then
$\Pr[\cT {\rm ~accepts~}] \le \frac13$

\end{itemize}
For a non-parametrized graph property $\cP$, an
\emph{$(\epsilon_1,\epsilon_2)$-tolerant}
tester is defined similarly by viewing $\cP$ as the single member of its parametrized
family and taking $\phi$ to be the identity function.

For a parametrized graph property, an \emph{$(\epsilon,\phi(\cdot))$-property
tester} is simply a $(0,\epsilon,\phi(\cdot))$-tolerant tester, and for a
non-parametrized graph property, an \emph{$\epsilon$-property tester}
is defined analogously.
\end{defn}

\section{Local reconstructors and tolerant testers}\label{s:results}

We now show that our notion of local reconstructors can be used alongside property
testers to construct tolerant testers for properties of sparse graphs. This idea
is not new and can be found as \cite[Theorem~3.1]{B08}, but we extend
the result for parametrized properties.

\begin{theorem}\label{t:tolerant_param}
Let $\cP_s$ be a parametrized graph property with an
$(\epsilon_1,\epsilon_2,\delta,\phi(\cdot))$-local reconstructor $\cR$
with query complexity $q_{\cR}$ and suppose
$\cP_{\phi(s)}$ has a $(\epsilon',\psi(\cdot))$-property tester $\cT$ with
query complexity $q_{\cT}$. Then
for all $\beta > 0$, $\cP_s$ has an $(\epsilon_1,\epsilon_2 + \epsilon' + \beta,
(\psi \circ \phi)(\cdot))$-tolerant tester with query complexity
$O\left( (1/\beta^2 + q_{\cT})q_{\cR}\right)$.
\end{theorem}

\later{
\begin{proof}
The algorithm and proof follows that of \cite[Theorem~3.1]{B08}.

The tolerant tester for $\cP_s$ is as follows:
\begin{enumerate}
\item%
Run $\cR$ on $G$ and estimate $\dist(G,\widetilde{G})$ to within additive
error of $\beta/2$ by sampling $(u,v) \in V \times V$.

\item%
If estimate of $\dist(G,\widetilde{G})$ exceeds $\epsilon_2 + \beta/2$, \textbf{reject}.

\item%
Run $\cT$ on $\widetilde{G}$ using $\cR$, and \textbf{accept} if and only if
$\cT$ accepts.

\end{enumerate}
If $G$ is $\epsilon_1$-close to $\cP_s$, then with high probability
$\dist(G,\widetilde{G}) \le \epsilon_2$. Therefore, the algorithm passes step 2
and with high probability $\cT$ accepts $\widetilde{G}$, since
$\widetilde{G} \in \cP_{\phi(s)}$.
If $G$ is $(\epsilon_2 + \epsilon' + \beta)$-far from $\cP_{\psi(\phi(s))}$, then
either $\dist(G,\widetilde{G}) > \epsilon_2 + \beta/2$, in which case step 2 fails
with (constant) high probability, or $\dist(G,\widetilde{G}) \le \epsilon_2 + \beta$ in which
case $\widetilde{G}$ is $\epsilon'$-far from $\cP_{\psi(\phi(s))}$ and so $\cT$ rejects
with high probability.
\end{proof}
}

Taking $\phi$ and $\psi$ to be identity, one gets as a special case the
result of~\cite[Theorem~3.1]{B08} that if $\cP$ is a graph property
with an $(\epsilon_1,\epsilon_2,\delta)$-local reconstructor and
an $\epsilon'$-property tester, then for all $\beta > 0$ it has an
$(\epsilon_1,\epsilon_2+\epsilon'+\beta)$-tolerant tester.



In this work, we give local reconstructors for several graph properties: connectivity
in undirected graphs, strong connectivity in directed graphs, and
small diameter in undirected graphs. To be precise, we prove
the following in Sections~\ref{s:conn},~\ref{s:sconn},~\ref{s:kconn} and~\ref{s:diam}
respectively.

\begin{theorem}\label{t:conn_recon}
There is an $(\epsilon, (1+\alpha)\epsilon, \delta)$-LR
for connectivity with query complexity~$O\left( \frac{1}{\delta\alpha\epsilon}\right)$.
\end{theorem}

\begin{theorem}\label{t:sconn_recon}
There is an $(\epsilon,(4+\alpha)\epsilon,\delta)$-LR
for strong connectivity
with query complexity~$O\left( \frac{1}{\delta\alpha\epsilon} \right)$.
\end{theorem}

\begin{theorem}\label{t:kconn_recon}
There is an $(\epsilon, (2+\alpha)\epsilon + ck/2, k(\delta+\gamma))$-LR
for $k$-connectivity
with query complexity
$O\left(
\left( \left(\frac{1}{c} + 1 \right)k\right)^{k}
t^{3k} (t+k)^{k} \log^k(t+k)\log^k(Cn)
\right)$
for $n \ge \frac{k}{c}$,
where $C = \frac{1}{\ln(1/(1-\gamma))}$
and $t = \frac{\ln(Cn)}{\delta\alpha\epsilon}$.
\end{theorem}

\begin{theorem}\label{t:diam_recon}
There is an
$(\epsilon,(3+\alpha)\epsilon+\frac1m+c,\delta+\frac1n, \phi(s)=2s+2)$-LR
for diameter at most $D$
with query complexity
$O(\frac{1}{c \delta\alpha\epsilon}\Delta^{O(\Delta\log\Delta)}\log n)$
where $\Delta = (\overline{d}/\epsilon)^{O(1/\epsilon)}$ and $\overline{d} = 2m/n$
is the bound on average degree.
\end{theorem}


Combining Theorem~\ref{t:tolerant_param} with each of Theorems~\ref{t:conn_recon},
\ref{t:sconn_recon}, \ref{t:kconn_recon}, and~\ref{t:diam_recon}, along with property testers
for each of the four properties (see~\cite{GR02}, \cite{BR02}, \cite{MR09}, and \cite{PR02}),
we immediately obtain the following tolerant testers.

\begin{cor}\label{c:conn_tol}
For all $\alpha, \beta, \epsilon > 0$, there is
an $(\epsilon, (1+\alpha)\epsilon + \beta)$-tolerant tester
for connectivity.
\end{cor}

\begin{cor}\label{c:sconn_tol}
For all $\alpha, \beta, \epsilon > 0$,  there is
an $(\epsilon, (4+\alpha)\epsilon + \beta)$-tolerant tester
for strong connectivity.
\end{cor}

\begin{cor}\label{c:kconn_tol}
For all $\alpha, \beta, c, \epsilon > 0$, there is an
$(\epsilon, (2k+\alpha)\epsilon + ck/2 + \beta)$-tolerant tester
for $k$-connectivity.
\end{cor}

\begin{cor}\label{c:diam_tol}
For all $D,\alpha,\beta,\epsilon > 0$ and constant $c < 1$, there is an
$(\epsilon, (3+\alpha)\epsilon + \frac1m + c + \beta, 4D+6)$-tolerant tester
for diameter at most $D$, for $n \ge \frac{k}{c}$.
\end{cor}

\section{Local reconstruction of connectivity}\label{s:conn}

In this section we prove Theorem~\ref{t:conn_recon}. We begin by giving
the high level description of the algorithm and then we present
the implementation and analysis of the algorithm.

\subsection{High level description}

The basic idea behind the algorithm is as follows. We designate
a ``super-node'' $v_0$ and add edges from a few special vertices to $v_0$
so that the resulting graph is connected. Ideally, we have exactly one
special vertex in each connected component, since this number of edges
is both necessary and sufficient to make the graph connected.
Therefore, we reduce the problem to defining a notion of ``special'' that can
be determined quickly, that ensures that at least one node per component
is special and that likely not too many extra nodes per
component are special.
How does a given
vertex know whether it is special? Our algorithm tosses coins to randomly
assign a rank $r(v) \in (0,1]$ to each $v \in V(G)$. Then $v$ can explore its
connected component by performing a breadth-first search (BFS) and if $v$
happens to have the lowest rank among all vertices encountered, then $v$ is special.
The only problem with this approach is that if $v$ lies in a large connected component,
then the algorithm makes too many queries to $G$ to determine whether $v$ is special.
We fix this by limiting the BFS to $K$ vertices, where $K$ is a constant depending
only on a few parameters, such as success probability and closeness.
We then show that components larger than size $K$ do not contribute many more
special vertices.

\later{

\iffull{
\subsection{Algorithm}
}\fi

We now give the algorithm. We first do some preprocessing.
In particular, we arbitrarily fix some vertex $v_0 \in V(G)$. Additionally,
we have a random oracle which, for each vertex $v \in V(G)$, it assigns
a random number $r(v) \in (0,1]$.

\begin{framed}
\begin{algorithmic}
\Procedure{Connected}{$v_1,v_2$}
\If{$E_G(v_1,v_2) = 1$}		\Comment{if edge is already in graph, it stays in the graph}
	\State \Return $1$
\Else
	\If{$v_0 \notin \{v_1,v_2\}$}  \Comment{do not add edge if neither endpoint is $v_0$}
		\State \Return $0$
	\Else	
		\State Let $v \in \{v_1,v_2\} \setminus \{v_0\}$
		\State BFS from $v$ up to $K$ vertices and let $U$ be the set
		of vertices visited
		\If{$r(v) < r(u)$ for all $u \in U$}	\Comment{check if $v$ is special}
			\State \Return $1$
		\Else
			\State \Return $0$
		\EndIf
	\EndIf
\EndIf
\EndProcedure
\end{algorithmic}
\end{framed}

The following lemma shows that the procedure does not add too many extra edges.
To prove it, we use the fact that if $G$ has at least $\epsilon m+2$ connected components,
then $G$ is $\epsilon$-far from being connected.

\begin{lemma}\label{l:connanal}
With probability at least $1 - \delta$, $\textsc{Connected}$ adds at most
$n/\delta K + \epsilon m + 1$ edges.
\end{lemma}
\begin{proof}
Let $X$ be the number of edges added by the local reconstructor.
Call a connected component $C$ \emph{small} if $|C| < K$ and \emph{large} otherwise.
Let $Y$ be the number of edges contributed by large components. Each small component
contributes exactly one edge, hence $X = i + Y$ where $i$ is the number of small components.
Since $G$ is $\epsilon$-close to being connected, $i \le \epsilon m +1$. Moreover,
each vertex in a large component contributes an edge with probability at most $1/K$.
By Markov's inequality,
$$
\Pr\left[Y > \frac{n}{\delta K}\right] \le \Pr\left[ Y > \frac{E[Y]}{\delta} \right] \le \delta. 
$$
hence
$$
\Pr\left[ X > \frac{n}{\delta K} + \epsilon m + 1 \right]
\le \Pr \left[ Y > \frac{n}{\delta K} \right] \le \delta.
$$
\end{proof}


\begin{proof}[Proof of Theorem~\ref{t:conn_recon}]
Let $\cP$ be the family of connected graphs on $n$ vertices.
The local reconstructor $\cR$ will run \textsc{Connected} with
$K = \frac{m}{\delta\alpha\epsilon m-1} = O\left( \frac{1}{\delta\alpha\epsilon} \right)$.
Clearly $\cR$ has query complexity $O(K) = O\left( \frac{1}{\delta\alpha\epsilon} \right)$.
To see that $\widetilde{G}$ is connected, observe that, on each connected component,
the rank function attains a minimum on some vertex, and that vertex therefore must get
an edge to $v_0$.
Furthermore, if $G$ is $\epsilon$-close to $\cP$, then by Lemma~\ref{l:connanal}
with probability at least $1-\delta$ the procedure will add
no more than $(1+\alpha)\epsilon m$ edges and hence $\dist(G,\widetilde{G}) \le
(1+\alpha)\epsilon$.
\end{proof}
}

\section{Local reconstruction of strong connectivity}\label{s:sconn}

In this section we prove Theorem~\ref{t:sconn_recon}. We first go over
some preliminary definitions and properties of directed graphs. Then
we give the high level description of the algorithm.
\iffull{
Finally we end
by presenting the implementation and analysis of the algorithm.
}\fi

\iffull{
\subsection{Preliminaries}
Throughout this section, we will use \emph{arc} to mean \emph{directed edge}.
}\fi

\later{
\begin{defn}
A directed graph $G$ is \emph{connected} if it is connected when viewing
arcs as undirected edges, and it is \emph{strongly connected} if there
is a path between every ordered pair of vertices.
A \emph{connected component} of $G$ is a maximal connected subgraph of
$G$, and a \emph{strongly connected component} is a maximal
strongly connected subgraph of $G$.
\end{defn}

\begin{defn}
A vertex is a \emph{source} (\emph{sink}) if it has no incoming (outgoing)
arcs. A strongly connected component with no incoming (outgoing) arcs is a
\emph{source (sink) component}.
\end{defn}
}

\iffull{
\paragraph{Query model.} In our model, we assume our neighbor oracle
has access to both the in-neighbor set and the out-neighbor set.
This allows us to perform
both backward and forward depth-first search (DFS), as well as undirected
BFS, which is a BFS ignoring directions of edges.
}\fi

\subsection{High level description}
The basic idea behind the algorithm is as follows. As in the undirected connectivity
case, we designate a ``super-node'' $v_0$, but now we add arcs from a few special
``transmitting'' vertices to $v_0$ and also add arcs from $v_0$ to a few special
``receiving'' vertices.
In order to make $G$ strongly connected, we need to add at least one arc
from the super-node to each source component and from each sink component
to the super-node without adding too many extra arcs. A na\"ive approach is to emulate
the strategy for connectivity: to decide if $v$ is a transmitter, do a forward DFS
from $v$ and check if $v$ has minimal rank (and analogously for receivers and backward
DFS). Again, we can limit the search so that large components may have some extra
special vertices. The problem with this approach is that a sink component could be
extremely small (e.g. one vertex) with many vertices whose only outgoing arcs lead
to the sink. In this case, all of these vertices would be special and receive an edge
to $v_0$. Therefore we tweak our algorithm so that if $v$ does a forward DFS and sees
few vertices, then it checks if it is actually in a sink component. If so, then it is a transmitter;
if not, then we do a limited \emph{undirected BFS} from $v$ and check minimality of rank.
\iffull{
We then show that we do not add too many extra edges this way.
}\fi

\iffull{
\subsection{Algorithm}
}\fi
\later{

We now present the algorithm. We do the same preprocessing as in
\textsc{Connected}, i.e. we arbitrarily fix a vertex $v_0 \in V(G)$ and
have access to a random oracle that randomly assigns ranks $r(v) \in (0,1]$
to each $v \in V(G)$.

\begin{framed}
\begin{algorithmic}
\Procedure{StronglyConnected}{$v_1,v_2$}
\If{$E_G(v_1,v_2) = 1$}		\Comment{if edge is already in graph, it stays in the graph}
	\State \Return $1$
\Else
	\If{$v_0 \notin \{v_1,v_2\}$}  \Comment{do not add arc if neither endpoint is $v_0$}
		\State \Return $0$
	\ElsIf{$v_2 = v_0$}	\Comment{check if $v_1$ is a transmitter}
		\State Forward DFS from $v_1$ up to $K$ vertices
		\If{Forward DFS sees at least $K$ vertices or
		\textsc{InSmallSink}$(v_1)$}
			\State \Return $1$ if $v_1$ has lowest
			rank among the DFS vertices else \textbf{return} $0$
		\Else
			\State Undirected BFS up to $K$ vertices
			\State \Return $1$ if $v_1$ has lowest
			rank among the BFS vertices else \textbf{return} $0$
		\EndIf
	\ElsIf{$v_1 = v_0$}	\Comment{check if $v_2$ is a receiver}
		\State Backward DFS from $v_2$ up to $K$ vertices
		\If{Backward DFS sees at least $K$ vertices or
		\textsc{InSmallSource}$(v_2)$}
			\State \Return $1$ if $v_2$ has lowest
			rank among the DFS vertices, else \textbf{return} $0$
		\Else
			\State Undirected BFS up to $K$ vertices
			\State \Return $1$ if $v_2$ has lowest
			rank among the BFS vertices, else \textbf{return} $0$
		\EndIf
	\EndIf
\EndIf
\EndProcedure
\end{algorithmic}
\end{framed}

The procedure uses two subroutines: \textsc{InSmallSource}$(v)$ and
\textsc{InSmallSink}$(v)$, which return True if $v$ is in a source (respectively
sink) component of size less than $K$. We will implement \textsc{InSmallSink}
(\textsc{InSmallSource} is similar except reverse all the directions
of edges) by running Tarjan's algorithm for finding strongly connected components,
except stopping after only exploring all nodes reachable from $v$.
If only one strongly connected component is returned, then return True,
otherwise return False. This clearly runs in $O(K)$ time since Tarjan's algorithm
runs in linear time and we are simply restricting the algorithm to the subgraph
induced by all strongly connected components reachable from $v$.

The following three lemmas capture the fact that if a directed graph is
almost strongly connected, then it cannot have too many source, sink, or
connected components, and therefore our algorithm likely does
not add too many arcs.

\begin{lemma}\label{c:scs}
If $G$ is $\epsilon$-close to being strongly connected, then $G$ has at most
$\epsilon m$ source components and at most $\epsilon m$ sink components.
\end{lemma}
\begin{proof}
Consider the directed graph $\widehat G$ of strongly connected components of $G$.
To make $G$ strongly connected, $\widehat G$ must have no sources or sinks,
yet adding an arc eliminates at most one source and at most one sink from
$\widehat G$.
\end{proof}

\begin{lemma}\label{l:scc}
If $G$ is $\epsilon$-close to being strongly connected, then $G$
has at most $\epsilon m + 1$ connected components.
\end{lemma}

\begin{proof}
If $G$ is $\epsilon$-close to being strongly connected, then it is also
$\epsilon$-close to being connected.
\end{proof}

\begin{lemma}\label{l:strongconnanal}
With probability at least $1-\delta$, the procedure will add no more
than $\frac{2n}{\delta K} + 4\epsilon n + 2$~arcs.
\end{lemma}
\begin{proof}
The analysis is similar to that of Lemma~\ref{l:strongconnanal}, except slightly more
complicated. Let $X$ be the random variable equal to the number of arcs
added. Let $S$ be the random variable equal to the number of arcs added that
end at $v_0$, and let $T$ be the random variable equal to the number of arcs
added that start at $v_0$. We will focus on $S$, since the analysis for $T$
is symmetrical. Any vertex $v$ for which \textsc{InSmallSink}$(v)$ is True
lies in a sink component of size less than $K$; call these components
small sink components. Also, any vertex that does
an undirected BFS and sees less than $K$ vertices must belong to a connected
component of size less than $K$; call these components small connected components.
Let $S'$ be the number of arcs counted by $S$ contributed by vertices not
in small sink components or small connected components. Then
$S \le S' + 2\epsilon m+1$ since each small sink component (of which there are
at most $\epsilon m$ by Corollary~\ref{c:scs}) contributes at most $1$ outgoing arc
and each small connected component (of which there are at most
$\epsilon m + 1$ by Lemma~\ref{l:scc}) contributes at most $1$ outgoing arc.
Define $T'$ analogously to $S'$, except for source components instead of sink
components, so that $T \le T' + 2\epsilon m + 1$. Then we have
$E[S'] \le \frac{n}{K}$ and $E[T'] \le \frac{n}{K}$ and therefore
\begin{eqnarray*}
\Pr\left[X > \frac{2n}{\delta K} + 4\epsilon m + 2\right] &\le&
\Pr\left[S'+T' > \frac{2n}{\delta K} \right] \\
&\le& \Pr\left[S'+T' > \frac{E[S'+T']}{\delta} \right] \\
&\le& \delta
\end{eqnarray*}
where the final inequality follows from Markov's inequality.
\end{proof}


\begin{proof}[Proof of Theorem~\ref{t:sconn_recon}]
The local reconstructor $\cR$ runs \textsc{StronglyConnected} with
$K = \frac{m}{\delta\alpha\epsilon m/2 -1}$.
To see that $\widetilde{G}$ is strongly connected, note that every source component
has some vertex of minimal rank, and this vertex will have the lowest
rank among its backward $K$-neighborhood, hence every source component
gets an arc from $v_0$. By a similar argument, every sink component
gets an arc to $v_0$. Now we must show that with high probability we did not add
too many edges. By Lemma~\ref{l:strongconnanal},
with probability at least $1-\delta$ we added at most $(4+\alpha)\epsilon m$ edges
and hence $\dist(G,\widetilde{G}) \le (4+\alpha)\epsilon$.
\end{proof}
}
\section{Implementing a neighbor oracle with connectivity reconstructor}\label{s:neighbor}

Our local reconstructors for $k$-connectivity and small diameter rely on the given graph $G$
being connected. Even if $G$ is not connected, it is close to being connected,
and so one may hope to first make $G$ into an intermediate connected graph $G'$
using a local reconstructor for connectivity, and then run the local reconstructor for the
desired property on $G'$ to obtain $\widetilde{G}$.
We would therefore like a neighbor oracle for $G'$, but the local reconstructor
only gives us an edge oracle for $G'$.
We show how to modify the local reconstructor for connectivity to obtain a neighbor
oracle for $G'$, with only a slight loss in the parameters achieved.

As is, our connectivity reconstructor $\textsc{Connected}$ from Section~\ref{s:conn} is
\emph{almost} a
neighbor oracle. Recall that the reconstructor works by selecting an arbitrary $v_0 \in G$
and adding edges from a few special vertices to $v_0$---no other edges are added.
For a vertex $v \ne v_0$,
$N_G(v) \subseteq N_{G'}(v) \subseteq N_G(v) \cup \{v_0\}$, thus
$N_{G'}(v)$ can be computed in constant time.
However, the problem arises when one queries $N_{G'}(v_0)$.
Potentially $\Theta(n)$ edges are added to $v_0$ by the reconstructor, so computing
$N_{G'}(v)$ queries $\textsc{Connected}$ $O(n)$ times. 
We thus modify $\textsc{Connected}$ to obtain the following.

\begin{theorem}\label{t:neighbor}
Fix a positive constant $c < 1$.
There is a randomized algorithm $N$, given access to the neighbor oracle
$N_G$ for $G$ that is $\epsilon$-close to being connected such that,
with probability at least $1-\delta$, there exists a connected graph $G'$
that is $((1+\alpha)\epsilon + c)$-close to $G$ and $N$ is a neighbor
oracle for $G'$, with query complexity $O\left( \frac{1}{c\delta\alpha\epsilon} \right)$.
\end{theorem}
\begin{proof}
We modify $\textsc{Connected}$ as follows. Instead of designating one super-node $v_0$,
we designate $c \cdot n$ super-nodes. Partition $V$ into sets of size $1/c$ and assign
each set in the partition to a distinct super-node. This can be implemented, for instance,
by identifying $V = \{1,\ldots,n\}$, designating the super-nodes to be
$\{1,\ldots,cn\}$, and for a given vertex $v \in V$, assign $v$ to the super-node
$h(v) =\lceil v/c \rceil$. For any $v$ that would be connected to $v_0$, we instead
connect it to $h(v)$. Additionally, we add the edges $(i,i+1)$ for all
$i \in \{1,\ldots,cn-1\}$ to ensure connectivity. This adds a total of $c-1$ edges,
which constitute at most $c$-fraction of the edges. Call this modified local
reconstructor $\textsc{Mod-Connected}$. It is straightforward to see that
$\textsc{Mod-Connected}$ has the same query complexity as $\textsc{Connected}$.
We now implement an algorithm to compute $N_{G'}$ as follows.
Given a non-super-node~$v$, its neighbor set could have grown by at most 
adding $h(v)$, so $N_{G'}(v)$ can be computed with $O(1)$ calls
to $\textsc{Mod-Connected}$. For a super node $w$, its neighbor set could have grown
by at most $1/c + 2$, since at most $1/c$ non-super-nodes could have been connected
to $w$, and $w$ is further connected to at most two super-nodes. Therefore
$N_{G'}(w)$ can be computed with $O(1/c)$ calls to $\textsc{Mod-Connected}$.
\end{proof}

\section{Local reconstruction of $k$-connectivity}\label{s:kconn}

\later{
}

In this section we prove Theorem~\ref{t:kconn_recon}. We first go over some preliminary
definitions and concepts related to $k$-connectivity. Then we give the high level
description of the algorithm, and finally we end by presenting the implementation of
the algorithm. All graphs in this section are undirected unless otherwise specified.
Throughout this section, we assume $k > 1$.

\subsection{Preliminaries}

For a subset $U \subsetneq V$ of the vertices in a graph $G = (V,E)$, the
\emph{degree of $U$}, denoted $\deg(U)$, is equal to
$
\deg(U) = |\{(u,v) \in E \mid u \in U, v \in V \setminus U \}|.
$

\begin{defn}[$k$-connectivity]
An undirected graph $G = (V,E)$ is \emph{$k$-connected} if for every
$U \subsetneq V$, $\deg(U) \ge k$.
\end{defn}

An equivalent definition of $k$-connectivity, a result of Menger's theorem
(see~\cite{Menger}), is that every pair of vertices has at least $k$ edge-disjoint
paths connecting them.
An important notion in the context of $k$-connectivity is that of an extreme set.
Extreme sets are a generalization of connected components to the $k$-connectivity
setting. Connected components are $0$-extreme sets.

\begin{defn}
A set $U \subseteq V$ is $\ell$-extreme if $\deg(U) = \ell$ and
$\deg(W) > \ell$ for every $W \subsetneq U$.
\end{defn}

It is straightforward from the definition that if a graph is $(k-1)$-connected and has
no $(k-1)$-extreme sets, then it is in fact $k$-connected.
Extreme sets satisfy some nice properties, which are used
by~\cite{MR09} as well for property testing and distance approximation
for $k$-connectivity.
Two extreme sets are either
disjoint or one is contained in the other (see~\cite{NGM97}. If $W \subsetneq U$ and $W$
is $\ell_W$-extreme and $U$ is $\ell_U$-extreme, then $\ell_W > \ell_U$.
Consequently, distinct $\ell$-extreme sets are disjoint.

One may hope there is some relationship between distance from $k$-connectivity
and the number of extreme sets, analogous to the relationship between distance from
connectivity and the number connected components.
Indeed, for a $(k-1)$-connected graph, there is. A graph $G$ that is
$(k-1)$-connected cannot have any $\ell$-extreme sets for $\ell < k-1$.
Moreover, the number of additional edges required to make $G$
$k$-connected is at least half the number of $(k-1)$-extreme sets in $G$.
This is simply because each $(k-1)$-extreme set requires at least one additional
edge, and adding an edge to $G$ meets the demand of at most two such sets.

\subsection{High level description}

The idea behind the algorithm is to simply iteratively make the graph
$j$-connected, for $j = 1,2,\ldots,k$. Let $G_j$ be the corrected $j$-connected
graph obtained from $G$. It suffices to use a neighbor oracle for $G_{k-1}$
to implement a neighbor oracle for $G_k$. The base case $k=1$ is addressed
by Section~\ref{s:neighbor}.

Now suppose we have a neighbor oracle for $(k-1)$-connectivity and
we wish to implement a neighbor oracle for $k$-connectivity.
Again, we use a similar idea as in Sections~\ref{s:conn} and~\ref{s:neighbor}.
Specifically, we fix a positive constant $k/n \le c < 1$ and designate a set $V_0 \subset V$
of $c \cdot n \ge k$
super-nodes, connecting them in a certain way to make the subgraph induced by $V_0$
$k$-connected (details in the next subsection).
The idea is then to ensure at least one vertex from each extreme set contributes
a new edge to a super-node.
Again, we implement this by assigning all vertices
a random rank independently and uniformly in $[0,1)$ and
searching a neighborhood of $v$
up to $t$ vertices, where $t$ is appropriately chosen,
and checking
if it has minimal rank.
Instead of doing this search via BFS, we use
the extreme set search algorithm of~\cite{GR02}.
The basic procedure satisfies the following: if $v$ lies in a $t$-bounded
extreme set, it finds this set with probability $\Theta(t^{-2})$, otherwise
it never succeeds.
We iterate the basic procedure
a polylogarithmic number of times.
If every iteration fails (which happens if $v$ does not lie in a $t$-bounded extreme set)
then we tell $v$ to connect to a super-node with probability $\Theta\left( \frac{\log (n/t)}{t} \right)$.
We then show that with high probability the resulting graph $G_k$ is $k$-connected and
that we do not add too many edges.
This completes the edge oracle. We will also show how implement these ideas
carefully so that the edge oracle can be transformed into a neighbor oracle
in order to make the recursion work.

\subsection{Algorithm}

Given our previous discussion, all that remains to implement the algorithm is to
implement the following tasks:
\begin{itemize}
\item%
\textbf{Search:}
Given $v \in V$ which lies in an extreme set $S$, find a neighborhood $U \subseteq S$
containing $v$ such that $|U| \le t$.
\item%
\textbf{Decision:}
Given $v \in V$, determine whether $v$ should contribute an edge to $V_0$,
and if so, to which $v' \in V_0$.
\end{itemize}

\subsubsection{Implementing the search task}
The goal of the search task is to detect that $v$ lies in an extreme set of
size at most $t$. We are now assuming the input graph is $(k-1)$-connected,
so all extreme sets are $(k-1)$-extreme.
The search task can be implemented by a method of~\cite{GR02} and~\cite{MR09}
which runs in
time $O(t^3 d \log(td))$ where $d$ is a degree bound on the graph. Roughly, the
procedure works by growing a set $U'$ starting with $\{v\}$ and iteratively choosing
a cut edge and adding the vertex on the other end of the edge into $U'$.
The cut edge is chosen by assigning random weights to the edges, and
choosing the edge with minimal weight. The procedure
stops when the cut size is less than $k$ or when $|U'|=t$.
\later{
The pseudocode for the extreme set search is as follows.

\begin{framed}
\begin{algorithmic}
\Procedure{ExtremeSetSearch}{$v,k-1$}
	\State For each edge, independently assign a random weight uniformly from $[0,1)$
	\State $U' \gets \{v\}$
	\Repeat
		\State $(u,w) \gets \arg \min \{{\rm wt}(u_1,u_2) \mid u_1 \in U', u_2 \notin U',
				(u_1,u_2) \in E(G_{k-1})\}$
			\Comment{Implicitly uses $(k-1)\textsc{-Conn}$}
		\State $U' \gets U' \cup \{w\}$
	\Until{$|U'| = t$ or $\deg(U') < k$}
\EndProcedure
\end{algorithmic}
\end{framed}
}

Its running time is $O(td\log(td))$ but
its success probability is only $\Omega(t^{-2})$ (that is, the probability
that $U' = S$ when the procedure terminates), so the basic procedure is repeated
$\Theta(t^2)$ times or until success. We can check if each run is successful by checking
that the final set $U'$ is a $(k-1)$-extreme set.
This procedure is adapted in~\cite{PR02}
to the general sparse model by noticing that no vertex of degree at least $t+k$
would ever be added to $S$, and hence the procedure has a running time of
$O(t^3(t+k)\log(t+k))$.
We actually want the probability that all vertices have a successful search
to be at least $1-\gamma/2$,
so we repeat the basic procedure $O(t^2 \log (Cn))$ times for
$C = \frac{1}{\ln(1/(1-\gamma/2))}$, yielding a time complexity
of $O(t^3(t+k)\log(t+k) \log (Cn))$.

\subsubsection{Implementing the decision task}
For the decision task, it is helpful to first think about how to do it globally.
First, we must hash each $v \in V$ to a set $h(v)$ of $k$ super-nodes in $V_0$.
This can be implemented as follows. Label $V = \{1,\ldots,n\}$ and
$V_0 = \{1,\ldots,cn\}$, and define
$
h(v) = \{\lceil v/c \rceil, \lceil v/c \rceil + 1, \ldots, \lceil v/c \rceil + (k-1)\}
$
where the entries are taken modulo $cn$. The specific way we do this hashing
is not important---it suffices to guarantee the following properties:
(1)
for every $v \in V$, $|h(v)| \ge k$
and
(2)
for every super-node $v' \in V_0$, there are at most a constant number, independent of $n$,
of $v$ such that $v' \in h(v)$, and that these $v$ are easily computable given $v'$;
our method guarantees the
constant $\frac{k}{c}$, which is the best one can hope for given Property 1.
Property 1 will be used later to ensure that the resulting graph is
$k$-connected (Lemma~\ref{l:kconn}).
Property 2 ensures, by the same reasoning as in
Section~\ref{s:neighbor}, that computing $N_{G_k}$ makes at most
$O\left(\frac{k}{c} + k\right)$
calls to $E_{G_k}$.

Now, to decide whether $v$ should be connected to a super-node, do an
extreme set search to find a neighborhood $U$, of size at most $t$, containing $v$,
and check if $v$ has minimal rank in $U$, where
the rank is the randomly assigned rank given in the high level description.
If so, then mark $v$ as \emph{successful}.
If the extreme set search fails, i.e. all $\Theta(t^2\log n)$ iterations fail, then
mark $v$ as successful
with probability $\frac{\ln(Cn/t)}{t}$ where $C = \frac{1}{\ln(1/(1-\gamma/2))}$ again,
which can be implemented by checking if the rank of $v$ is less than $\frac{\ln(Cn/t)}{t}$.
If $v$ is successful, then find the lexicographically smallest
super-node $v' \in h(v)$ to which $v$ is not already connected.
If none exist, then do nothing; otherwise, connect $v$ to $v'$.

We also want the subgraph induced by $V_0$ to be $k$-connected to help ensure
that $G_k$ is $k$-connected (Lemma~\ref{l:kconn}). Globally,
from each $i \in V_0$ we add an edge to $i+1,i+2,\ldots,i+\lceil k/2 \rceil$, taken
modulo $cn$. This ensures that the subgraph induced by $V_0$ is
$k$-connected (Lemma~\ref{l:supernode_kconn}). Locally, this is implemented as follows.
If $(i,j) \in [cn]^2$ is queried, if the edge is already in $G_{k-1}$, then the local reconstructor
returns $1$, otherwise it returns $1$ if and only if $j \in \{i+1,i+2,\ldots,i+\lceil k/2 \rceil
\pmod{cn}\}$ or $i \in \{j+1,j+2,\ldots,j+\lceil k/2 \rceil \pmod{cn}\}$.

\later{
\subsubsection{Pseudocode for $k$-connectivity given $(k-1)$-connectivity}
The pseudocode for the $k$-connectivity reconstructor is given below. It assumes
the input graph $G_{k-1}$ is $(k-1)$-connected.

\begin{framed}
\begin{algorithmic}
\Procedure{$k$-Conn}{$v_1,v_2$}
	\If{$(k-1)\textsc{-Conn}(v_1,v_2)=1$}
	\Comment{if edge is already in graph, it stays in the graph}
		\State \Return $1$
	\Else
		\If{$v_1 \notin V_0$ and $v_2 \notin V_0$}
			\Comment{do not add edge if neither endpoint is in $V_0$}
			\State \Return $0$
		\ElsIf{$v_1,v_2 \in V_0$}
			\If{$v_1 \in \{v_2 \pm 1, \ldots, v_2 \pm \lceil k/2 \rceil \pmod{cn}\}$}
				\State \Return $1$
			\Else
				\State \Return $0$
			\EndIf
		\Else
			\State Let $v \in \{v_1,v_2\} \setminus V_0$
			\State Let $v' \in \{v_1,v_2\} \setminus \{v\}$
			\Comment{Decide if $v$ should be connected to $v' \in V_0$}
			\State $b \gets \textsc{false}$ \Comment{connect $v$ if $b$ is true}
			\For{$i=1,\ldots,t^2\log(Cn)$}
				\If{$\textsc{ExtremeSetSearch}(v,k-1)$ succeeds}
					\State $U \gets \textsc{ExtremeSetSearch}(v,k-1)$
					\If{$v$ has minimal rank in $U$}
						\State $b \gets \textsc{true}$
					\EndIf
					\State Break
				\EndIf
			\EndFor
			\If{$b$ is false and $r(v) < \frac{\ln(Cn/t)}{t}$}
				\State $b \gets \textsc{true}$
			\EndIf
			\If{$b$ is true and $v'$ is lexicographically smallest in $h(v) \setminus
					N_{G_{k-1}}(v)$}
				\State \Return $1$
			\Else
				\State \Return $0$
			\EndIf
		\EndIf
	\EndIf
\EndProcedure
\end{algorithmic}
\end{framed}
}

\iffull{
\subsubsection{Analysis}
}\fi
\later{

The following lemma ensures the subgraph induced by $V_0$ after adding these
edges is $k$-connected.

\begin{lemma}\label{l:supernode_kconn}
For $n \ge 2k+1$, let $V = \{0,1,\ldots,n-1\}$, $n \ge 2k+1$, and
$
E = \{(i,i+j \pmod{n}) \mid i \in V, 1 \le j \le k\}.
$
Then $G = (V,E)$ is $2k$-connected.
\end{lemma}
\begin{proof}
We claim it suffices to show that for every $i$, there are $2k$ edge-disjoint paths
connecting $i$ to $i+1 \pmod{n}$. Suppose this is true. Consider any cut $(C, V \setminus C)$
of the vertices. To show that this cut has $2k$ cut edges, it suffices to show that there exist
$i \in C$, $j \in V \setminus C$ with $2k$ edge-disjoint paths connecting $i$ to $j$.
But there exists $i$ such that $i \in C$ and $i+1 \in V \setminus C$. Hence our claim
implies the assertion.

It remains to prove our claim. By symmetry, we may assume $i=0$. We will
exhibit $2k$ explicit edge-disjoint paths from $0$ to $1$. For $j = 2,\ldots,k$,
we have paths which traverse edges $(0,j), (j,j+1), (j+1,1)$ as well as
paths which traverse edges $(0,n-j), (n-j,n-j+1), (n-j+1, 1)$. This accounts
for $2k-2$ paths which are pairwise edge-disjoint. Note that, in these
$2k-2$ paths, the only edges of the form $(a,a+1)$ used are for
$a \in \{2,\ldots,k\} \cup \{n-k,\ldots,n-2\}$, and the only edges
for the form $(b,b+k)$ used are for $b \in \{0,1,n-k,n-k+1\}$. The final two paths are
the path which is simply the edge $(0,1)$, and the path which traverses edges
$(0,n-1), (n-1,n-k-1), (n-k-1,n-k-2), (n-k-2,n-k-3), \ldots, (k+3,k+2), (k+2,2), (2,1)$.
\end{proof}

The following lemma ensures that the resulting graph is $k$-connected.

\begin{lemma}\label{l:kconn}
With probability at least $1 - \gamma$, the resulting graph $G_k$ is $k$-connected.
\end{lemma}
\begin{proof}
We first show that with probability at least $1-\gamma$, every extreme set of $G_{k-1}$
has a vertex which adds an edge to $V_0$. First, consider each small extreme set,
i.e. extreme sets of size at most $t$. For each such set $U$, there is some $v_U \in U$
with minimal rank in $U$.
By the choice of $C$, with probability at least
$1-\gamma/2$ every $v_U$ has a successful extreme set search and recognizes that
it has minimal rank in $U$, hence every small extreme set contributes an edge to $V_0$.
Next, consider each large extreme set, i.e. extreme sets of size more than $t$.
The probability that a large extreme set contributes no edges, i.e.
has no successful vertices, is at most $\left(1 - \frac{\ln(Cn/t)}{t} \right)^t < \frac{t}{Cn}$.
Since there are at most $n/t$ large extreme sets, the probability that every
large extreme set contributes an edge is at least $1-\gamma/2$ by the choice of $C$.
Hence with probability at least $1-\gamma$ every extreme set contributes an edge
to $V_0$.

Now we assume that every extreme set contributes an edge to $V_0$ and
show that $G_k$ is $k$-connected.
It suffices to show that, for every extreme set $S$ in $G_{k-1}$, one of the edges added by
the procedure crosses between $S$ and $V \setminus S$. We have three cases.
\begin{itemize}
\item%
$V_0 \subseteq V \setminus S$:
Consider $v \in S$ of minimal rank. With high probability, the neighborhood $U$
of $v$ found by the search algorithm is contained in $S$, so $v$ is also minimally
ranked in $U$. There must exist $v^* \in h(v)$ to which $v$ is not already connected,
for if not, then $v$ is connected to $|h(v)| \ge k$ vertices outside of $S$, contradicting
that $S$ is $(k-1)$-extreme.
\item%
$V_0 \subseteq S$:
We claim there is another extreme set $S' \ne S$ in $G_{k-1}$.
Suppose not.
Then for every subset $S'' \subseteq V \setminus S$,
$\deg(S'') \ge k$, for otherwise $V \setminus S$ contains an extreme set.
But $\deg(V \setminus S) = \deg(S) = k-1$, so $V \setminus S$ is an extreme set,
a contradiction.
Now, recall that $S'$ is disjoint from $S$.
By the same argument as in the previous case, with high probability
there is $v \in S'$ which gets connected to some super-node in $V_0 \subseteq S$.
\item%
$V_0 \cap S$ and $V_0 \setminus S$ are both non-empty:
We claim that in $G_k$, $S$ has degree $\deg(S) \ge k$.
Consider the cut $(V_0 \cap S, V_0 \setminus S)$
within the subgraph $V_0$. Any cut edge here must be a cut edge in
the cut $(S, V \setminus S)$. By Lemma~\ref{l:supernode_kconn},
the subgraph $V_0$ is $k$-connected, so the cut $(V_0 \cap S, V_0 \setminus S)$,
and hence the cut $(S, V \setminus S)$, has at least $k$ cut edges.
\end{itemize}
\end{proof}

The following lemma
shows that with high probability not too many additional edges are added.

\begin{lemma}\label{l:dist}
With probability at least $1-\delta$, the number of edges added is at most
$$
2\epsilon m + ckn/2 + \frac{n \ln(Cn/t)}{\delta t}.
$$
\end{lemma}
\begin{proof}
Each $t$-bounded extreme set will contribute at most one edge.
There are at most $2\epsilon m$ extreme sets.
On average $1/t$ of the remaining vertices contribute edges.
Finally, the super-nodes themselves contribute at most $ckn/2$ edges.
Therefore we add at most $2\epsilon m + ckn/2 + X$ edges, where
$X$ is a random variable with mean
$\mathbb{E}[X] \le \frac{n\ln(Cn/t)}{t}$.
Markov's inequality implies
$$
\Pr \left[X > \frac{\mathbb{E}[X]}{\delta} \right] \le \delta.
$$
\end{proof}

\begin{proof}[Proof of Theorem~\ref{t:kconn_recon}]
Set $C = \frac{1}{\ln(1/(1-\gamma))}$ and $t = \frac{\ln(Cn)}{\delta\alpha\epsilon}$.
By Lemma~\ref{l:kconn} our resulting graph is $k$-connected with probability
at least $1-\gamma$ and by
Lemma~\ref{l:dist}, $\dist(G_{k-1}, G_k) \le (2+\alpha)\epsilon + ck/2$ with
probability $1-\delta$
and therefore by induction $\dist(G, G_k) \le (2+\alpha)k\epsilon + ck^2/2$.
This can be improved to $(2+\alpha)k\epsilon + ck/2$ by noting that the
same super-nodes and the same edges between them can be used for
all intermediate graphs $G_1,\ldots,G_{k-1}$. The success probability
is at least $(1-\delta)^k(1-\gamma)^k \ge 1- k(\delta+\gamma)$.
The query complexity for correcting
$G_{k-1}$ to $G_k$ is $O(t^3(t+k)\log(t+k)\log(Cn))$ queries to $N_{G_{k-1}}$.
But each call to $N_{G_{k-1}}$ takes $O\left(\left(\frac{1}{c} +1\right)k \right)$
calls to $E_{G_{k-1}}$.
By induction, the query complexity of $E_{G_k}$ is
$O\left( \left( \left(\frac{1}{c} + 1\right) k \right)^{k} \cdot t^{3k}(t+k)^{k} \log^k(t+k)
\log^k(Cn)\right)$.
\end{proof}
}
\section{Local reconstruction of small diameter}\label{s:diam}

\later{
}

In this section we prove Theorem~\ref{t:diam_recon}. We first go over
some preliminaries on graph diameter. We then give the high level
description of the algorithm and prove some useful characteristics of graphs
that are close to having small diameter before finally presenting the
implementation and analysis of the algorithm.

\subsection{Preliminaries}

\begin{defn}
Let $G$ be a graph with adjacency matrix $A$. For an integer $k$, let $G^k$
be the graph on the same vertex set as $G$, with adjacency matrix $A^k$
(boolean arithmetic).
\end{defn}

It is not hard to show that there is an edge between $u$ and $v$ in $G^k$ if and only if
there is a path of length at most $k$ between $u$ and $v$ in $G$.
This observation immediately gives us the following.

\begin{prop}\label{p:chardiam}
Let $G$ be a graph and let $D > 0$ be an integer. Then
$\diam G \le D$ if and only if $G^D$ is a complete graph.
\end{prop}

\subsection{High level description}

The basic idea behind the algorithm is as follows. Again, we designate a ``super-node''
$v_0$ and add edges between $v_0$ and a few special vertices.
If the input graph is close to having diameter at most $D$, then we aim for our
reconstructed graph to have diameter at most $2D+2$.
We show that if $G$ is close to having diameter at most $D$, then we have an upper bound
on the size of any independent set in $G^D$ (Lemma~\ref{l:addedge} and
Corollary~\ref{c:close}).

Ideally, then, we want our special vertices (those that get an edge to $v_0$) to
be a dominating set in $G^D$ that is not too large. If we add edges from the dominating
set in $G^D$ to $v_0$, then
our resulting graph $\widetilde{G}$ has diameter at most $2D+2$, since
any vertex can reach some vertex in the dominating set
within $D$ steps, and hence $v_0$ within $D+1$ steps.
If this dominating
set is a maximal independent set, then
our upper bound on the size of independent sets in $G^D$ also upper bounds
the number of edges we add. However, all
known algorithms for locally computing a maximal independent set
have query complexity bounded in terms of the maximum degree of the
graph. A variant of the algorithm found in~\cite{NO08} has been analyzed
by~\cite{YYI09} and~\cite{ORRR12}
to run in expected time bounded in terms of the average degree of the graph,
but this average is taken over not only coin tosses of the algorithm
but also all possible queries. This is undesirable for us since we want a uniform
bound on the query complexity for any potential query vertex, given a ``good''
set of coin tosses. We want an algorithm such that, for most sets of coin tosses,
we get the correct answer \emph{everywhere}, whereas the algorithm of
\cite{NO08} leaves open the possibility of failure for some queries, regardless
of the coin tosses. If some queries give the wrong answer, then the fact
that our reconstructed graph retains the property is compromised.
Instead, we do something less optimal but good enough.
Note that adding edges does not increase diameter. Instead of using a maximal
independent set, we settle for a dominating set in $G^D$ as long as we can
still control its size. 

\subsection{Properties of graphs close to having small diameter}

The following lemma and subsequent corollary state that if a graph
$G$ is close to having diameter at most $D$, then no independent set
in $G^D$ can be very large.

\later{
}
\both{
\begin{lemma}\label{l:addedge}
Let $v_1,\ldots,v_k$ be an independent set in $G^D$ and let
$H$ be the subgraph of $G^D$ induced by this set. Let $(s,t)$ be
an edge not in $E(G)$ and let $G' = (V(G), E(G) \cup \{(s,t)\})$ be the
graph obtained by adding $(s,t)$ to $G$. Let $H'$ be the subgraph
of $(G')^D$ induced by $v_1,\ldots,v_k$. Then, for some $i \in \{1,\ldots,k\}$,
all edges in $H'$ (if any) are incident to $v_i$. In particular,
if $G^D$ has an independent set of size $k$, then $(G')^D$ has an
independent set of size $k-1$.
\end{lemma}
}
\later{
\begin{proof}
We will prove the assertion as follows. First, we will show that
any two edges in $H'$ must share a vertex. Then, we show that any third
edge must also be incident to that vertex. This implies that all edges
share a common vertex.

Suppose that the edges $(v_i, v_j)$ and $(v_k, v_l)$ are in $H'$.
Then without loss of generality the paths from $v_i$ to $v_j$
and from $v_k$ to $v_l$ in $G'$ must use $(s,t)$. Let $a = d(v_i, s)$, $b = d(t,v_j)$,
$c = d(v_k, s)$ and $d = d(t, v_l)$, so we have
\begin{eqnarray*}
a + b + 1 &\le D \\
c + d + 1 &\le D.
\end{eqnarray*}
Adding yields $a+b+c+d \le 2D-2$, so at least one of $(a+c)$, $(b+d)$ must
be less than or equal to $D$. Without loss of generality, suppose $a+c \le D$.
This implies that $d(v_i,v_k) \le D$ even in $G$, so it must be
that $v_i = v_k$.

Now suppose yet another edge is in $H'$. By what we just showed,
it must share an incident vertex with $(v_i, v_j)$, and it must do likewise with
$(v_i, v_l)$. So either it is incident to $v_i$, in which case we are done, or
the edge is $(v_j, v_l)$. We will show the latter cannot happen. Define
$a, b, d$ as above. Note that we have
$$
a + d + 1 \le D.
$$
Suppose there is path of length $\le D$ from $v_j$ to $v_l$ in $G'$ which traverses
$(s,t)$ by entering via $s$ and exiting via $t$. Let $e = d(v_j, s)$. Then we have
$$
e + d + 1 \le D.
$$
However, since any path from $v_i$ to $v_j$ that does not traverse $(s,t)$ must
have length greater than $D$, we have $a + e > D$ from which we can deduce
$$
e > b+1.
$$
Similarly, since any path from $v_j$ to $v_l$ that does not traverse $(s,t)$ must
have length greater than $D$, we have
$$
b + d > D.
$$
But then $D \ge e + d + 1 > b + d + 2 > D + 2$, a contradiction.
\end{proof}
}

\begin{cor}\label{c:close}
Suppose $G$ is $\epsilon$-close to having diameter $\le D$. Then any
independent set in $G^D$ has size at most $\epsilon m + 1$.
\end{cor}
\iffull{
\begin{proof}
Suppose $G^D$ has an independent set of size $\epsilon m + 2$. There
exist a set of $\epsilon m$ edges such that, when added to $G$,
we obtain $G'$ with $\diam(G') \le D$, and hence $(G')^D$ is a complete graph.
But, by Lemma~\ref{l:addedge}, $(G')^D$ has an independent set of size
$2$, a contradiction.
\end{proof}
}\fi

Finally, we will use a result from~\cite{AGR, PR02} which we will restate in our own
terms below (the theorem number we reference is from~\cite{AGR}):

\begin{theorem}[{\cite[Theorem~3.1]{AGR}}]
Any connected graph $G$ is $(\frac{2n}{Dm})$-close to having diameter~$\le D$.
\end{theorem}

This result implies that any connected graph is $\epsilon$-close to having
diameter at most $\frac{2n}{\epsilon m}$, and will come in handy when we
want to bound our query complexity. In particular, the query complexity of our algorithm
will depend on $D$, but if $D \ge \frac{2n}{\epsilon m}$, which is $O(1)$
for constant $\epsilon$, then we can instead simply aim for diameter
$\frac{2n}{\epsilon m}$ and bound our query complexity in terms of this constant
instead.

\subsection{Algorithm}
\later{
We start by fixing a super-node $v_0 \in G$.
Given our discussion in the high level description, it remains to
implement the selection of a small dominating set.
To this end, we create a dominating set $S$ by first adding
the set $H$ of high-degree vertices into $S$ and then using the local
maximal independent set algorithm found in~\cite{RTVX11}, which is
based on Luby's algorithm (\cite{Luby}), on $G^D$ with $H$ and its vertices'
neighbors (in $G^D$) removed to create an independent set $M$. Then let
$S = H \cup M$.

By high-degree vertex we mean a vertex with degree greater than
$\overline{d}/\epsilon$, of which there are at most $\epsilon n$, where
$\overline{d} = \frac{2m}{n}$ is a bound on the average degree.
We will also consider the super-node $v_0$ a high-degree vertex for our purposes.
To implement choosing a maximal independent set, we define
a subroutine $\textsc{MIS}_D(v)$ as follows. On input $v$, simulate the
local maximal independent set described in~\cite{RTVX11}, except explore
all neighbors within $D$ steps from $v$ rather than just immediate neighbors,
and automatically reject if the algorithm encounters any vertex with degree
exceeding $\overline{d}/\epsilon$.
Note that the running time of
the local MIS algorithm given in~\cite{RTVX11} is bounded in terms of
the maximum degree of the graph. This is not problematic since our variant
of the algorithm ignores any vertex with degree greater than $\overline{d}/\epsilon$,
hence the effective maximum degree of our graph $G$ for
the purposes of the algorithm is $\overline{d}/\epsilon$, so the effective
maximum degree of $G^D$ is $(\overline{d}/\epsilon)^D$.

One final challenge is that if $D$ is large, say $\Theta(\log n)$, then
our query complexity bound in terms of our effective degree is no longer sublinear.
We work around this by using a result of~\cite{AGR}, which states that
every connected graph is $\epsilon$-close to having diameter at most
$\frac{2n}{\epsilon m} = O(1/\epsilon)$. Therefore, we can aim for
achieving diameter $K = \min\{D, \frac{2n}{\epsilon m}\}$ so that our effective degree
is $(\overline{d}/\epsilon)^K = (\overline{d}/\epsilon)^{O(1/\epsilon)}$.
Of course, this only works if our graph is connected to begin with. Therefore,
we use the neighbor oracle for the connected correction $G'$ of $G$,
given in Section~\ref{s:neighbor}. The idea is then to first make $G$ into a connected graph
$G'$, and then reconstruct a small diameter graph $\widetilde{G}$ out of $G'$.\\
}

\later{
\begin{framed}
\begin{algorithmic}
\Procedure{SmallDiam}{$v_1,v_2,D$}
\State $K \gets \min\{\frac{2n}{\epsilon m}, D\}$
\If{$E_{G'}(v_1,v_2) = 1$}		\Comment{if edge is already in graph, it stays in the graph}
	\State \Return $1$
\Else
	\If{$v_0 \notin \{v_1,v_2\}$}  \Comment{do not add edge if neither endpoint is $v_0$}
		\State \Return $0$
	\Else
		\State Let $v \in \{v_1,v_2\} \minus \{v_0\}$ \Comment{check if $v$ is special}
		\If{$\deg_{G'}(v) > \overline{d}/\epsilon$} 
			\State \Return $1$
		\Else
			\State \Return $\textsc{MIS}_K(v)$  \Comment{neighbor queries
			to $G'$ instead of $G$}
		\EndIf
	\EndIf
\EndIf
\EndProcedure
\end{algorithmic}
\end{framed}
}

\later{
\begin{proof}[Proof of Theorem~\ref{t:diam_recon}]
Correcting $G$ to $G'$ with $\textsc{Mod-Connect}$
adds $(1+\alpha)\epsilon m + cm$ edges.
The query complexity is clearly dominated by that of $\textsc{MIS}_K(v)$.
The probability of success follows from that of the local MIS algorithm.
Now we show correctness. Let $H$ be the set of vertices in $G$ with degree
greater than $\overline{d}/\epsilon$, and let $M$ be the set of vertices for which
$\textsc{MIS}(v) = 1$. Observe that $M$ is an independent set in $G^D$,
so $|M| \le \epsilon m$. Note that since we forced $v_0$ to be in
this union, the number of edges we added is actually
$
|H \cup M \minus \{v_0\}| = |H| + |M| - 1 \le \epsilon n + \epsilon m \le 2\epsilon m.
$
Furthermore, let $v \in V(G)$. There exists some $u \in H \cup M$ such that
$d(v,u) \le K$, otherwise $v$ would have been added to $M$. Therefore,
$d(v,v_0) \le K+1$. For any other $v' \in V(G)$, we thus have
$
d(v,v') \le d(v,v_0) + d(v_0,v') \le 2K+2 \le 2D+2.
$
\end{proof}
}



\appendix
\magicappendix

\end{document}